\documentclass[reqno,11pt]{imsart}
\usepackage{amsfonts,amsmath,amsthm,mathtools,bbm,bm,eucal}
\usepackage{times,amssymb,mathabx} 
\usepackage[notrig]{physics}
\usepackage{graphicx,caption,subcaption,hyperref,enumitem,multirow} 
\usepackage[ruled,vlined,linesnumbered]{algorithm2e}
\usepackage[nameinlink,capitalize,noabbrev]{cleveref}
\usepackage{crossreftools,multibib} 
\usepackage[authoryear]{natbib}
\usepackage[title,toc]{appendix}
\usepackage[dvipsnames,x11names]{xcolor}
\usepackage{tikz}

\usepackage{booktabs}
\usepackage{amsmath}
\usepackage{graphicx}

\newcites{S}{References}
\hypersetup{colorlinks=True, urlcolor=Blue, citecolor=MidnightBlue, linkcolor=RedViolet}

\numberwithin{equation}{section}

\newtheorem{theorem}{Theorem}
\newtheorem{lemma}{Lemma}
\newtheorem{proposition}{Proposition}

\theoremstyle{definition}

\newtheorem{assumption}{Assumption}

\newtheorem{examplex}{Example}
  {\popQED\endexamplex}
\newtheorem{remarkx}{Remark}
\newenvironment{remark}
  {\pushQED{\qed}\remarkx}
  {\popQED\endremarkx}

\crefname{assumption}{Assumption}{Assumptions}
\crefname{examplex}{Example}{Examples}
\crefname{theorem}{Theorem}{Theorems}
\crefname{lemma}{Lemma}{Lemmas}
\crefname{section}{Section}{Sections}
\crefname{algorithm}{Algorithm}{Algorithms}
\crefname{enumi}{Assumption}{Assumptions}
\crefname{equation}{Inequality}{Inequalities}
\Crefname{equation}{Equation}{Equations}

\makeatletter
\newcommand*{\rom}[1]{\expandafter\@slowromancap\romannumeral #1@}
\makeatother

\newcommand{\genericdel}[4]{%
  \ifcase#3\relax
  \ifx#1.\else#1\fi#4\ifx#2.\else#2\fi\or
  \bigl#1#4\bigr#2\or
  \Bigl#1#4\Bigr#2\or
  \biggl#1#4\biggr#2\or
  \Biggl#1#4\Biggr#2\else
  \left#1#4\right#2\fi
}
\newcommand{\del}[2][-1]{\genericdel(){#1}{#2}}
\newcommand{\cbr}[2][-1]{\genericdel\{\}{#1}{#2}}

 \let\abs\Abs
 \let\norm\Norm
\newcommand{\floor}[2][-1]{\genericdel\lfloor\rfloor{#1}{#2}}

\def\d{\textup{d}}

\def\E{\textsf{\textup{E}}}
\def\P{\textsf{\textup{P}}}
\def\Q{\textsf{\textup{Q}}}
\def\R{\mathbb{R}}

\DeclareMathOperator*{\argmin}{\arg\!\min}
\DeclareMathOperator*{\argmax}{\arg\!\max}

\def\ind{\mathbbm{1}}

\def\var{\textup{Var}}

\def\Fin{F_{\textsc{in}}}
\def\Fout{F_{\textsc{out}}}
\def\hFin{\widehat{F}_{\textsc{in},t+1}}
\def\hFout{\widehat{F}_{\textsc{out},t+1}}
\def\bcG{\bar{\mathcal{G}}}
\def\hg{\hat{g}}
\def\hcR{\widehat{\mathcal{R}}}
\def\hcE{\widehat{\mathcal{E}}}

\makeatletter
\def\greekvectors#1{%
 \@for\next:=#1\do{%
    \def\X##1;{\expandafter\def\csname b##1\endcsname{\bm{\csname##1\endcsname}}}
    \expandafter\X\next;}
 \@for\next:=#1\do{%
    \def\X##1;{\expandafter\def\csname h##1\endcsname{\widehat{\csname##1\endcsname}}}
    \expandafter\X\next;}
 \@for\next:=#1\do{%
    \def\X##1;{\expandafter\def\csname t##1\endcsname{\widetilde{\csname##1\endcsname}}}
    \expandafter\X\next;}
 \@for\next:=#1\do{%
    \def\X##1;{\expandafter\def\csname ba##1\endcsname{\bar{\csname##1\endcsname}}}
    \expandafter\X\next;}
 \@for\next:=#1\do{%
    \def\X##1;{\expandafter\def\csname c##1\endcsname{\check{\csname##1\endcsname}}}
    \expandafter\X\next;}
 \@for\next:=#1\do{%
    \def\X##1;{\expandafter\def\csname u##1\endcsname{\underline{\csname##1\endcsname}}}
    \expandafter\X\next;}
 \@for\next:=#1\do{%
    \def\X##1;{\expandafter\def\csname hb##1\endcsname{\widehat{\bm{\csname##1\endcsname}}}}
    \expandafter\X\next;}
}
\greekvectors{alpha,beta,gamma,delta,epsilon,zeta,theta,kappa,lambda,mu,nu,xi,pi,rho,tau,sigma,phi,chi,psi,omega,varepsilon,varpi,varphi,vartheta,  Delta,Gamma,Theta,Lambda,Xi,Pi,Sigma,Phi,Psi,Omega}
    
\@tfor\next:=abfghijlmnopqrstuwxyzABCDGHIJKLMOSTUVWXYZ\do{%
    \def\command@factory#1{\expandafter\def\csname #1\endcsname{\mathbf{#1}} }
    \expandafter\command@factory\next
}
\@tfor\next:=abcdeghijklnopqrstuvwxyzABCDEFGHIJKLMNOQRSTUVWXYZ\do{%
    \def\command@factory#1{\expandafter\def\csname b#1\endcsname{\mathbbm{#1}} }
    \expandafter\command@factory\next
}
\@tfor\next:=abcdeghijklnpqrstuvwxyzABCDEFGHIJKLMNOPQRSTUVWXYZ\do{%
    \def\command@factory#1{\expandafter\def\csname t#1\endcsname{\texttt{\textup{#1}}} }
    \expandafter\command@factory\next
}
\@tfor\next:=dknoABCDEFGHIJKLMNOPQRSTUVWXYZ\do{%
    \def\command@factory#1{\expandafter\def\csname c#1\endcsname{\mathcal{#1}} }
    \expandafter\command@factory\next
}
\@tfor\next:=abdefghijklmnopqrstuvwxyzABCDEFGHIJKLMNOPQRSTUVWXYZ\do{%
    \def\command@factory#1{\expandafter\def\csname s#1\endcsname{\textsf{\textup{#1}}} }
    \expandafter\command@factory\next
}
\@tfor\next:=abcdefghjklmnopqrstuvwxyzABCDEFGHIJKLMNOPQRSTUVWXYZ\do{
    \def\command@factory#1{\expandafter\def\csname f#1\endcsname{\mathfrak{#1}} }
    \expandafter\command@factory\next
}
\@tfor\next:=abcdeghijklmnopqstuvwxyzABCDEFGHIJKLMNOPQRSTUVWXYZ\do{%
    \def\command@factory#1{\expandafter\def\csname ba#1\endcsname{\bar{#1}} }
    \expandafter\command@factory\next
}
\@tfor\next:=xyzABCDEFGHIJKLMNOPQRSTUVWXYZ\do{%
    \def\command@factory#1{\expandafter\def\csname h#1\endcsname{\widehat{#1}} }
    \expandafter\command@factory\next
}
\@tfor\next:=abcdefghijklmnopqrstuvwxyzABCDEFGHIJKLMNOPQRSTUVWXYZ\do{%
    \def\command@factory#1{\expandafter\def\csname ti#1\endcsname{\widetilde{#1}} }
    \expandafter\command@factory\next
}
\@tfor\next:=abcdefghjklmnopqrstuvwxyzABCDEFGHIJKLMNOPQRSTUVWXYZ\do{
    \def\command@factory#1{\expandafter\def\csname u#1\endcsname{\underline{#1}} }
    \expandafter\command@factory\next
}
\@tfor\next:=ABCDEFGHIJKLMNOPQRSTUVWXYZ\do{%
    \def\command@factory#1{\expandafter\def\csname hc#1\endcsname{\widehat{\mathcal{#1}}} }
    \expandafter\command@factory\next
}
\@tfor\next:=ABCDEFGHIJKLMNOPQRSTUVWXYZ\do{%
    \def\command@factory#1{\expandafter\def\csname uc#1\endcsname{\underline{\mathcal{#1}}} }
    \expandafter\command@factory\next
}

\usepackage[letterpaper, margin=1.2in]{geometry} 

\begin{document}

\begin{frontmatter}

\title{Rate-Optimal Neural Boundary Detection from Unlabeled Noisy Images}
\runtitle{Rate-Optimal Image Boundary Detection}

\begin{aug}
\author[A]{\fnms{Kyeongho} \snm{Kim}}
\and
\author[A]{\fnms{Ilsang} \snm{Ohn}\ead[label=e1,mark]{ilsang.ohn@inha.ac.kr}}
\address[A]{Department of  Statistics, Inha University}%
\printead{e1} 
\end{aug}
\runauthor{K. Kim and I. Ohn}

\begin{abstract}
We study boundary detection for unlabeled noisy images from a statistical perspective. The aim is to recover an unknown object region from raw intensity observations without pixel-wise annotating labels or a parametric model for the intensity distributions. Motivated by robust Gibbs posterior approaches based on thresholded misclassification losses, we propose a continuous hinge-type surrogate loss for boundary detection. The proposed loss is amenable to gradient-based optimization and can be combined with deep neural networks to represent complex object boundaries. We prove that the proposed loss function is Fisher consistent under a mild separation assumption and obtain a calibration inequality linking excess surrogate risk to the symmetric difference error of the estimated region. Under a piecewise smooth boundary model, we prove that the resulting deep neural network estimator achieves the minimax-optimal boundary recovery rate, up to logarithmic factors. The piecewise smooth formulation accommodates boundaries with corners and kinks, thereby extending beyond globally smooth boundary models. Numerical experiments demonstrate that the proposed method accurately and stably recovers object boundaries across a range of noise levels and shape configurations, and compares favorably with existing unsupervised boundary detection methods.
\end{abstract}



\end{frontmatter}

\section{Introduction}

Detecting object boundaries from noisy pixel intensities is a fundamental problem in image analysis and image segmentation, with applications ranging from medical imaging to remote sensing and industrial inspection. From a statistical viewpoint, the goal is to recover an unknown object region $\Gamma_\star \subset[0,1]^2$ (or equivalently its boundary) based on i.i.d.\ observations $\{(X_i,Y_i)\}_{i=1}^n$, where $X_i\in [0,1]^2$ denotes a pixel location and $Y_i\in\mathbb{R}$ is a noisy intensity measurement at $X_i$. A common modeling assumption is that $X$ is generated from some continuous distribution $\Q$ on $[0,1]^2$ and the conditional distribution of $Y$ depends on whether $X$ lies inside or outside the true region $\Gamma_\star\subset[0,1]^2$, namely,
    \begin{align}
  \label{eq:model_intro}
        X&\sim \Q\\
    Y|X &\sim \Fin \ind(X\in\Gamma_\star) +  \Fout\ind(X\in\Gamma_\star^\complement),  
    \end{align}
for some unknown distributions $\Fin$ and $\Fout$ on $\mathbb{R}$, where $\ind$ denotes the indicator function. While $\Fin$ and $\Fout$ may be of scientific interest in certain applications, the primary target in boundary detection is $\Gamma_\star$ itself, and it is desirable to develop methods that can accurately recover $\Gamma_\star$ without relying on a detailed parametric model for the intensity distributions.

To address this statistical boundary detection problem, it is natural to consider likelihood-based or Bayesian methods, which usually define likelihood models for both $\Fin$ and $\Fout$, along with specifying an appropriate model for the boundary. The theoretical optimality of these approaches was derived by  \citet{hall2001local,li2017bayesian}. However, while they can be effective when the likelihood is correctly specified, it may be sensitive to misspecification of the pixel intensity model, and it introduces additional nuisance parameters that can substantially increase computational cost. Motivated by these concerns, \citet{syring2020robust} proposed a robust Gibbs posterior approach that constructs a posterior distribution for the boundary directly through a loss function, avoiding explicit modeling of $\Fin$ and $\Fout$. They established that the resulting Gibbs posterior concentrates around the true boundary at a (near-)minimax optimal rate under a suitable smoothness condition on the boundary, demonstrating that likelihood-free boundary inference can be both robust and statistically efficient.

Despite these attractive theoretical properties, the Gibbs posterior formulation in \citet{syring2020robust} for boundary detection is tied to loss functions that are discontinuous and combinatorial in nature. The loss function they employed is based on a thresholded misclassification error, which leads to an objective function that is non-smooth in the boundary parameter. While this is compatible with an MCMC algorithm, it becomes challenging to scale to modern high-dimensional function classes and gradient-based optimization schemes that are now standard in large-scale image analysis. This motivates the following question:

\begin{quote}
   \normalsize{ \emph{Can we design a continuous, optimization-friendly loss for boundary detection that achieves rate-optimal recovery of $\Gamma_\star$ over practically relevant function classes such as deep neural networks?}}
\end{quote}

In this paper, we answer this question affirmatively by introducing a continuous surrogate loss for image boundary detection that serves as a hinge-type relaxation of the threshold-based loss. An important characteristic of the proposed loss, which we refer to as \textit{Fisher consistency} in analogy to classification problems \citep{lin2004note}, is its ability to recover the true object region, in the sense that the population risk associated with the proposed loss is  minimized at the true object boundary. It parallels the identifiability condition in \citet{syring2020robust}, but our continuous relaxation makes it possible to incorporate flexible function classes, including deep neural networks, in a theoretically grounded way.

The proposed loss-based framework can be viewed as reformulating boundary detection as a weighted classification problem, in which pseudo-labels are derived via intensity-based thresholding. We propose a novel \textit{adaptive loss calibration} scheme to adjust the class weights and threshold level during training in a fully data-driven manner. At a high level, the method iteratively refines the separation between candidate interior and exterior regions based on the current boundary estimate, producing a sequence of calibrated loss functions that progressively align with the underlying signal structure. This adaptive mechanism allows the procedure to operate without prior knowledge of the intensity distributions and stabilizes learning in the presence of unknown noise levels, while remaining compatible with gradient-based optimization.

On the theoretical side, we show that a deep neural network estimator trained with the proposed loss function achieves a rate-optimal convergence guarantee for boundary estimation. Notably, our framework accommodates piecewise smooth boundaries, including shapes with corners and kinks, thereby extending existing results that rely on global smoothness assumptions considered in \citet{li2017bayesian,syring2020robust}. Specifically, when the boundary consists of several pieces that have smoothness index $\beta>0$, our estimator satisfies
    \begin{align}
        \lambda(\widehat{\Gamma}_n \triangle \Gamma_\star)
    \lesssim n^{-\beta/(\beta+1)}
    \end{align}
up to a logarithmic factor with high probability, where $\lambda$ denotes the Lebesgue measure and $\triangle$ denotes the symmetric difference. This matches the classical minimax rate for boundary estimation in two dimensions \citep{mammen1995asymptotical} and demonstrates that continuous-loss-based neural boundary detection can be both computationally scalable and statistically optimal.

A large body of recent research on boundary detection employs deep learning methods because of their powerful representational capabilities \citep[e.g,][]{xie2015holistically, liu2017richer, xu2024ctbound}. They are typically trained in a supervised manner using pixel-wise labeled masks, thereby requiring substantial annotation effort to construct ground-truth segmentation labels. In contrast, the present work estimates object boundaries directly from \textit{raw} intensity measurements and can be applied in a fully unsupervised manner, providing a simple and practical alternative when annotating pixel-wise segmentation masks is often costly or infeasible such as medical imaging, remote sensing, and scientific imaging application.

The remainder of the paper is organized as follows. In the rest of this section, we provide a brief survey of related work and introduce our notation. In \cref{sec:method}, we introduce the proposed continuous loss and the corresponding empirical risk minimization procedure, along with the dynamic parameter update strategy. We also prove that the proposed loss function is Fisher consistent and provide a calibration inequality for a symmetric difference error. \cref{sec:theory} presents our main theoretical results on convergence rates. Numerical experiments are reported in \cref{sec:experiments}. \cref{sec:conclusion} concludes the paper.  Proofs and technical lemmas are given in \cref{sec:proofs}.

\subsection{Related work}

\paragraph{Boundary detection and image segmentation}
Image boundary detection is a fundamental preprocessing step for capturing structural features in images and has been extensively studied in computer vision, image processing, medical imaging, and industrial inspection, resulting in a large body of literature .Early approaches focused on local, gradient-based methods that detect edges by exploiting variations in pixel intensities, with classical examples including the Sobel, Prewitt, and Canny operators \citep{sobel1968isotropic, prewitt1970object, canny1986computational}. While these methods are computationally efficient and effective in low-noise settings, their performance deteriorates substantially in the presence of noise, where spurious gradients can obscure true boundaries. To address this limitation, subsequent work has incorporated denoising and regularization techniques, such as the K-SVD–based enhancement of Canny edge detection proposed by \citet{wei2022canny_sparse}. Beyond local gradient information, a broad class of methods has been developed to exploit global or region-based structure. These include variational and active contour models \citep{chan2001active}, multiscale contour detection frameworks \citep{arbelaez2011contour}, and information-theoretic approaches based on mutual information or hierarchical modeling \citep{isola2014crisp, ofir2020faint}. These methods improve robustness by incorporating spatial coherence and global context.

\paragraph{Deep learning approaches for boundary detection}
Recent advances in deep learning have led to a reformulation of boundary detection as a supervised learning problem based on neural networks. Early work by \citet{xie2015holistically, liu2017richer} developed convolutional neural network (CNN) architectures that exploit multiscale feature representations for pixel-wise boundary prediction. Building on this line of work, more recent approaches incorporate global contextual information through hybrid CNN–Transformer architectures, which have been shown to improve robustness in complex and noisy imaging environments \citep{xu2024ctbound}. In parallel, several studies have explored unsupervised or weakly supervised alternatives to reduce the reliance on labeled segmentation masks. For instance, \citet{lajili2025unsupervised} integrate the Ambrosio–Tortorelli variational energy into deep learning, enabling boundary estimation via energy minimization. \citet{wang2023cut} leverage objectness cues from self-supervised Vision Transformers to construct pseudo-labels, followed by iterative self-training with noise-robust loss functions.

\paragraph{Surrogate losses in classification}
In statistical learning, it is standard to replace discontinuous losses such as the $0$-$1$
misclassification loss by continuous surrogate losses (e.g., hinge or logistic losses) in order to
enable scalable optimization. For example, the exponential loss function is used by AdaBoost \citep{friedman2000additive} and the hinge loss is used by support vector machines (SVMs) \citep{steinwart2008support}.  The logistic loss, also known as cross-entropy, is commonly employed in modern machine learning approaches, including deep learning \citep{mao2023cross,zhang2024classification}. The use of such surrogate losses can be theoretically justified  by establishing a quantitative relationship between the risks evaluated with a surrogate and misclassification loss functions, which is called a calibration inequality \citep{zhang2004statistical,bartlett2006convexity}. From this perspective, the hinge loss is appealing because it yields a tight upper bound of the misclassification risk, thereby enabling us to achieve optimal convergence rates. This optimality characteristic of the hinge loss has been analyzed in the context of SVMs \citep{tarigan2006classifiers,blanchard2008statistical,steinwart2007fast} as well as deep neural networks \citep{kim2021fast}.

\paragraph{Statistical learning theory for deep neural networks}
Motivated by the empirical success of deep learning, a large literature has investigated the
theoretical properties of deep neural networks. Since a comprehensive review is beyond the scope
of this paper, we focus on results developed from the perspective of statistical learning theory.
In the context of nonparametric regression, deep neural networks have been shown to achieve
minimax optimal convergence rates over various function classes; see, for example,
\citet{schmidt2020nonparametric,kohler2021rate,imaizumi2020advantage, fang2024intrinsic, fan2024factor}. In particular, \citet{imaizumi2020advantage} demonstrates that deep networks can approximate and learn piecewise smooth functions efficiently, which is closely related to our piecewise smooth boundary model. For classification, several works have established optimality properties of deep neural network classifiers trained with convex surrogate losses. 
\citet{kim2021fast,zhang2025optimal} studied hinge loss based deep classifiers and derived optimal convergence rates under suitable margin and complexity conditions, while \citet{zhang2024classification,bos2022convergence,ohn2022nonconvex} investigated the theoretical properties of the deep classifiers minimizing the logistic loss.

\subsection{Notation}

For $x\in\R$, we let $(x)_+:=\max\{0,x\}$, which denotes the positive part of the real number $x$. Let $\floor{x}$ denote the largest integer not larger than $x$.  For a natural number $n\in\bN$, we denote $[n]:=\{1,\dots, n\}$.  For two positive sequences $(a_n)_{n\in \mathbb{N}}$ and $(b_n)_{n\in \mathbb{N}}$, we write $a_n\lesssim b_n$ or  $b_n\gtrsim a_n$, if there exists a positive constant $C>0$ such that $a_n\le Cb_n$ for any $n\in \mathbb{N}$. Moreover, we write $a_n\asymp b_n$ if both $a_n\lesssim b_n$  and  $a_n\gtrsim b_n$ hold. Let $\lambda$ denote the Lebesgue measure. For a function $f$ supported on $\cX$, let $\|f\|_{\cL_q(\cX)}:=(\int_{\cX}|f(x)|^q\d\lambda(x))^{1/q}$ denote the $\cL_q$ norm on $\cX$ for $q\in\bN$ and $\|f\|_{\cL_\infty(\cX)}:=\sup_{x\in\cX}|f(x)|$ the $\cL_\infty$ norm on $\cX$. Let $\cC^p(\cX)$ denote the class of $p$-times differentiable functions supported on $\cX$ for $p\in\bN$ and $\cC^0(\cX)$ the class of continuous functions on $\cX$.  For a set $A$, we denote by $A^\complement$ its complement and by $|A|$ its cardinality. For two sets $A$ and $B$, their symmetric difference is denoted by $A\triangle B:=(A\cap B^\complement)\cup (A^\complement \cap B)$.

\section{Methodology}
\label{sec:method}
\subsection{A new continuous loss function}

We represent the unknown region through a \textit{decision function} $g^\star:[0,1]^2\to\mathbb{R}$ as
\begin{equation}\label{eq:region_g_intro}
\Gamma_\star = \{x\in[0,1]^2: g^\star(x)\ge 0\}.
\end{equation}
Without loss of generality, we can set
    \begin{align}
    \label{eq:true_function}
        g_\star(x) = 2\ind (x\in\Gamma_\star)-1.
    \end{align}
In this paper, we estimate the true decision function $g^\star$ by minimizing empirical risk with a suitably designed loss function over a class of candidate functions (e.g., neural networks). 

Our starting point is the threshold-based misclassification loss considered in \citet{syring2020robust}, which is given by
 \begin{align*}
    \tilde{\ell}_g(y,x)&:=   \kappa\ind\del[1]{y>\xi, x\in\Gamma^\complement} +\tau\ind\del[1]{y\le \xi, x\in\Gamma}\\
    &=     \kappa\ind\del[1]{y>\xi, x<g(x)} +\tau\ind\del[1]{y\le \xi, g(x)\ge0}
    \end{align*}
for a threshold parameter $\xi\in\mathbb{R}$ and weighting parameters $\kappa>0$ and $\tau>0$, where the second equality follows from  our modeling approach for the object region $\Gamma= \{x\in[0,1]^2:g(x)\ge0\}$. 

As we mentioned in the introduction, this loss function is hard to optimize due to its discrete nature. Motivated by this observation, we replace the discontinuous indicator structure by a hinge-type continuous relaxation. Specifically, we define a ``pseudo-label''
    \begin{align}
    \label{eq:pseudo_lable}
        u_y := 2\ind(y>\xi)-1 \in \{-1,1\}
    \end{align}
and then we propose the loss
    \begin{align}
    \ell_g(y,x)
    &:= (1-u_yg(x))_+\{\kappa\ind(y>\xi)+\tau\ind(y\le \xi)\}\nonumber\\
     &= (1-u_yg(x))_+\{\kappa(u_y+1)/2+\tau(1-u_y)/2\},
     \label{eq:continuous_loss}
\end{align}
which can be viewed as a weighted hinge loss applied to the classification of the pseudo-label. Here, we drop the dependence of the loss parameters $\xi$, $\kappa$ and $\tau$ for notational simplicity. 
This loss is continuous in $g(x)$ and admits efficient optimization using standard gradient-based methods, while still retaining a direct connection to the boundary $\Gamma_\star$.

A key property of the proposed loss is that it is \emph{Fisher consistent} for boundary recovery under certain conditions on the loss parameters given in the next assumption.

\begin{assumption}[Loss function]
\label{assume:loss}
There exist loss parameter values $\xi$, $\kappa$ and $\tau$ satisfying
    \begin{align}
    \label{eq:tuning_cond}
        \Fin(\xi)< \frac{\kappa}{\kappa+\tau}< \Fout(\xi).
    \end{align}
\end{assumption}

\cref{assume:loss} requires that the inside and outside intensity distributions $\Fin$ and $\Fout$ should be separated at the threshold level $\xi$, in order to detect the boundary: if they are too close or identical, then no statistical procedure can identify the object region.

This condition is slightly weaker than Assumption A in \citet{syring2020robust}. This weaker requirement is not inherent to the loss itself; rather, it arises from our proof technique based on Bernstein's inequality, which contrasts with \citet{syring2020robust}, whose analysis relies on an exponential moment condition.

In the following proposition, we establish the Fisher consistency of the proposed loss function. For technical simplicity, we restrict our attention to the bounded function class $\bcG:=\{g: \|g\|_{\cL_\infty([0,1]^2)}\le 1\}$. This restriction is not essential since our target region is defined only through the sign of $g$, and in practice, it can be enforced by a simple truncation or normalization step without affecting the induced region.

\begin{proposition}[Fisher consistency]
\label{prop:fisher_consist}
Under \cref{assume:loss}, the function $g_\star$ in \eqref{eq:true_function} is the minimizer of the population risk evaluated with the loss function $\ell_g$, i.e., 
    \begin{align}
        g_\star \in \argmin_{g\in\bcG}\cR(g):=\E[ \ell_g(Y,X)].
    \end{align}
\end{proposition}

In the following proposition, we establish a ``calibration inequality'', which relates the boundary detection error and the excess risk evaluated on the proposed continuous loss function. We need one technical condition, which assumes that the density of the pixel location distribution is bounded away from 0 and above by a certain constant.

\begin{assumption}[Pixel distribution]
\label{assume:pixel}
The distribution $\Q$ of a pixel location $X$ satisfies
    \begin{align}
       \frac{1}{A}
       \le  \inf_{x\in[0,1]^2}\frac{\d\Q}{\d\lambda}(x)
      \le   \sup_{x\in[0,1]^2}\frac{\d\Q}{\d\lambda}(x)\le A
    \end{align}
for some absolute constant $A>1$.
\end{assumption}

\begin{proposition}[Calibration inequality]
\label{prop:calib}
Under \cref{assume:loss,assume:pixel}, there exists an absolute constant $\tC_0>0$ such that
    \begin{align}
        \cR(g)-\cR(g_\star)\ge \tC_0 \lambda(\Gamma \triangle \Gamma_\star)
    \end{align}
with $\Gamma=\{x:g(x)\ge 0\}$ for any function $g\in\bcG$.
\end{proposition}

In this paper, we propose to use an empirical risk minimizer
    \begin{align}
      \label{eq:estimator}
        \hg_n =\argmin_{g\in\cG_n}\sum_{i=1}^n \ell_g(Y_i,X_i),
    \end{align}
where $\cG_n$ is some class of functions used for estimation, e.g., neural networks. The corresponding estimator of the object region is then
    \begin{align}
        \hGamma_n=\cbr{x\in[0,1]^2:\hg_n(x)\ge0}.
    \end{align}
Although the proposed procedure is applicable to arbitrary function classes, we focus on deep neural networks in our theoretical analysis and numerical studies, since they offer both
strong approximation guarantees for complex decision boundaries and practical scalability via
efficient gradient-based optimization.

\subsection{Adaptive loss calibration}
\label{subsec:dynamic_update}

A practical challenge in loss-based boundary inference is to choose the loss parameters $(\xi,\kappa,\tau)$ to satisfy \cref{assume:loss}, since the distributions $\Fin$ and $\Fout$ are unknown. To address this, we develop an \emph{adaptive loss calibration} strategy that adaptively calibrates the loss parameters during optimization, following the loss scaling suggestion of \citet{syring2020robust}. At each iteration, we use the current estimate of $g$ to form empirical approximations of $\Fin$ and $\Fout$, and update $\xi$ to maximize the separation between the two groups, followed by an update of $(\kappa,\tau)$ to approximately enforce \eqref{eq:tuning_cond}. This yields an automatic, data-driven calibration procedure that is simple to implement and empirically stable.

Specifically, for a time-step $t\in\bN$, let $g_t$ denote the current estimate of the decision function. We update $g_{t+1}$ from $g_t$, typically using a gradient-based method to (approximately) minimize the empirical risk induced by the current loss. Based on $g_{t+1}$, we form empirical estimates of the distribution functions $\Fin$ and $\Fout$ by
    \begin{align}
        \hFin(\xi)&:=\frac{|\{i\in[n]:Y_i\le \xi, g_{t+1}(X_i)\ge 0\}|}{|\{i\in[n]:g_{t+1}(X_i)\ge 0\}|},
            \label{eq:cdf_est1}\\
         \hFout(\xi)&:=\frac{|\{i\in[n]:Y_i\le \xi, g_{t+1}(X_i)< 0\}|}{|\{i\in[n]:g_{t+1}(X_i)< 0\}|}    \label{eq:cdf_est2}
    \end{align}
We then update the threshold parameter $\xi_{t+1}$ by maximizing the separation between the two empirical
distribution functions as
    \begin{align}
    \label{eq:thresh_update}
        \xi_{t+1}\in \argmax_{\xi\in \R}\cbr{ \hFout(\xi) - \hFin(\xi)}.
    \end{align}
Next, we update the weighting parameters $(\kappa_{t+1},\tau_{t+1})$ as
        \begin{align}
         \label{eq:lossweight_update}
         \frac{\kappa_{t+1}}{\kappa_{t+1}+\tau_{t+1}} 
         =\frac{1}{2}\cbr{\hFin(\xi)+  \hFout(\xi)},
    \end{align}
so that the inequality \eqref{eq:tuning_cond} is satisfied by these estimated quantities. As the loss parameters are updated, the loss function is updated accordingly. For each $i\in[n]$, we update the pseudo-label as
    \begin{align}
        U_{i,t+1}:=2\ind(Y_i>\xi_{t+1})-1
    \end{align}
and define the updated loss by
    \begin{align}
    \label{eq:loss_update}
       \ell_{g,t+1}(Y_i, X_i):= (1-U_{i,t+1}g(X_i))_+\cbr{\kappa_{t+1}(U_{i,t+1}+1)/2+\tau_{t+1}(1-U_{i,t+1})/2}.
    \end{align}

\begin{remark}
The update formula for $\kappa$ and $\tau$ given in \eqref{eq:lossweight_update} differs from that of \citet{syring2020robust}, which sets $\kappa_{t+1}/(\kappa_{t+1}+\tau_{t+1})=\widehat{F}(\xi_{t+1}):=n^{-1}|\{i\in[n]:Y_i\le \xi_{t+1}\}|,$ i.e., the ratio depends on the marginal empirical distribution of the intensity measurements. But this update may be unstable because the estimate $\widehat{F}(\xi_{t+1})$, which is a weighted average of $\hFin(\xi_{t+1})$ and $\hFout(\xi_{t+1})$  with weights proportional to the number of pixels in their respective regions, can be dominated by the larger region. This may lead to overly imbalanced weights and, consequently, degenerate iterates where the learned decision function has the same sign for most pixels. To address this numerical instability, we instead apply the update rule \eqref{eq:lossweight_update}, which relies on their simple average with equal weights.
\end{remark}

\subsection{Summary of algorithm}

In \cref{alg:main}, we provide a complete description of the proposed algorithm used to compute the decision function estimator in \eqref{eq:estimator}, incorporating the dynamic update strategy for the loss parameters.

\begin{algorithm}[H]
\caption{Image boundary detection via continuous loss minimization}
\KwIn{Data $\{(X_i,Y_i)\}_{i\in[n]}$, initial function $g_1$, initial loss parameters $(\xi_1,\kappa_1,\tau_1)$, maximum iterations $T$.}

\KwOut{Estimated decision function $\hg_n=g_T$ and region estimator $\hGamma_n=\{x\in[0,1]^2:\hg_n(x)\ge 0\}$.}
\For {$t=1,2,\dots, T-1$} {
    Update $g_{t+1}$ from $g_t$ using a gradient-based method to (approximately) minimize the empirical risk induced by the current loss;\\
    Compute $\hFin$ and $\hFout$ according to \eqref{eq:cdf_est1}--\eqref{eq:cdf_est2} using $g_{t+1}$;\\    
    Set $\xi_{t+1}$ to satisfy \eqref{eq:thresh_update}; \\    
    Choose $(\kappa_{t+1},\tau_{t+1})$ to satisfy \eqref{eq:lossweight_update};\\    
    Update the loss function according to \eqref{eq:loss_update}.
} 
\label{alg:main} 
\end{algorithm}

\section{Theory}
\label{sec:theory} 

In this section, we demonstrate that applying the proposed loss-based minimization method with deep neural networks can yield an object region estimator that achieves an optimal convergence rate.

\subsection{Neural networks}

We introduce our notation for neural network models. For a positive integer $L\in\bN$ larger than 1 and a $(L+1)$-dimensional vector of positive integers $m_{0:L}:=(m_0, m_1,\dots, m_{L})\in\bN^{L+1}$, we denote $\Theta(m_{0:L}; B):=\bigotimes_{l=1}^{L}([-B,B]^{m_{l}\times m_{l-1}}\times [-B,B]^{m_{l}})$, where $B>0$ is a magnitude bound. For a \textit{network parameter} $\theta=((W_l, b_l))_{l\in[L]}\in\Theta(m_{0:L}; B)$, we define the \textit{deep  neural network} $\tilde{g}(\cdot|\theta)$ induced by the network parameter $\theta$ as
    \begin{align*}
         \tilde{g}(x|\theta)=[W_{L},b_{L}]\circ\rho\circ[W_{L-1},b_{L-1}]\circ\cdots\circ\rho\circ[W_{1},b_{1}]x, 
    \end{align*}
where $[W_{k},b_{k}]$ denotes the affine transformation represented as a multiplication by the weight matrix $W_k$ and an addition of the bias vector $b_k$, i.e., $[W_{k},b_{k}]x = W_kx + b_k$, and $\rho$ does the elementwise ReLU (rectified linear unit) activation function $\rho(x)=((x_j)_+)_{j}$. To retain the Fisher consistency in \cref{prop:fisher_consist}, we truncate the output of a neural network at level $1$ as
    \begin{align*}
        g(x|\theta)=\min\{\max\{\tilde{g}(x|\theta),-1\},1\}.
    \end{align*}
We model the decision function $g_\star$ by a class of neural networks defined as
    \begin{align*}
        \cG(L, M, B):=\cbr{g(\cdot|\theta):\theta\in\Theta(m_{0:L};B)\text{ with }m_0=2,\max_{1\le l\le L-1}m_l\le M, m_L=1},
    \end{align*}
which is a set of deep neural networks with  depth $L$, width $M$, magnitude bound $B$ and truncation level $1$.

\subsection{Piecewise smooth boundaries}

We assume that the boundary of the true object region consists of several pieces that are H\"older smooth. Let $\cH^\beta([0,1],K)$ denote the H\"older ball defined as 
\begin{equation*}
	\cH^\beta([0,1], r):=\cbr{g\in\cC^{\floor{\beta}}([0,1]):\|g\|_{\cH^\beta([0,1])}\le r},
	\end{equation*}
where $\|\cdot\|_{\cH^\beta([0,1])}$ denotes the H\"older norm defined as
    \begin{align*}
    \|g\|_{\cH^\beta([0,1])}
        :=\max\cbr{\max_{k\in\bN_0:k<\beta}\norm[2]{\frac{\d^{k}g}{\d x^k}}_{\cL_\infty([0,1])}
        ,\sup_{x_1,x_2\in [0,1]: x_1\neq x_2 }\frac{|\partial^{\floor{\beta}}g(x_1)-\partial^{ \floor{\beta}}g(x_2)|}{|x_1-x_2|^{\beta- \floor{\beta}}}}.
    \end{align*}

\begin{assumption}[Object region]
\label{assume:true}
The true object region is given by
         \begin{align}
         \label{eq:true_region}
        \Gamma_\star = \bigcap_{j=1}^J \cK_{j}
    \end{align}
for some fixed $J\in\bN$,  where each $K_j$ is a half-space-like region represented as
    \begin{align*}
        \cK_j \in \cbr[2]{
        \{x: x_1\ge h_j(x_2)\},
        \{x: x_1\le h_j(x_2)\}, 
        \{x: x_2\ge h_j(x_1)\},
        \{x: x_2\le h_j(x_1)\}
    },
    \end{align*}
with a boundary function $h_j\in\cH^\beta([0,1],r)$ for some $\beta>0$ and $K>0$.
\end{assumption}

Our piecewise smooth boundary assumption generalizes the globally smooth boundary condition considered in \citet{li2017bayesian,syring2020robust}. For example, a rectangular region $\Gamma^\star=[a,b]\times[c,d]\subset[0,1]^2$ has a boundary consisting of four infinitely smooth line segments, but it is not globally $\cC^1$ due to the presence of corners. Such shapes are excluded under global smoothness assumptions, while they are naturally included in our framework since $\Gamma_\star$ can be written as the intersection of half-space-like sets with H\"older boundary functions (in this case, constants). Similarly, regions with kinked boundaries can be represented using piecewise smooth functions (e.g., piecewise linear), which are excluded under global smoothness assumptions.

\subsection{Convergence rate}

The first step to derive the optimal convergence property of the proposed estimator is to show that deep neural networks can closely approximate the true decision function.

\begin{theorem}[Approximation]
\label{thm:approximation}
Suppose that \cref{assume:true} is made.
There exist absolute positive constants $L_0$, $\tC_1$, $\tC_2$ and $\tC_3$ such that for any $D\in\bN\setminus\{1\}$  there exists a neural network $g^\dag\in\cG(L_0, \tC_1 D, \tC_1D^{\tC_2})$ that satisfies
    \begin{align}
        \|g^\dag-g_\star\|_{\cL_1([0,1]^2)}\le \tC_3 D^{-2\beta}.
    \end{align}
\end{theorem}

With \cref{thm:approximation} in hand, we derive the convergence rate of the proposed region estimator with deep neural networks by employing the standard concentration argument for empirical processes.

\begin{theorem}[Convergence rate]
\label{thm:conv}
Under \cref{assume:loss,assume:pixel,assume:true}, the empirical risk minimizer $\hg_n$ given in \eqref{eq:estimator}  with $\cG_n=\cG(L_0, M_n, B_n)$ leads to the estimator $\hGamma_n$ of the object region such that
    \begin{align}
        \lambda(\hGamma_n \triangle\Gamma_\star)\le \tC_1\cbr{n^{-\beta/(\beta+1)}\log n +\frac{\log (1/\delta)}{n}}
    \end{align}
with probability at least $1-\delta$ when $\tC_2 n^{1/(2\beta+2)}\le M_n \le \tC_3 n^{1/(2\beta+2)} $ and $ n^{\tC_4}\le B_n\le n^{\tC_5}$ for some absolute positive constant  $L_0$ and $\tC_1,\dots, \tC_5$.
\end{theorem}

The convergence rate in the above theorem matches the minimax rate for boundary estimation in two dimensions \citep{mammen1995asymptotical}, up to a logarithmic factor.

\begin{remark}
\cref{thm:conv} is derived for fixed loss parameters $(\xi,\kappa,\tau)$ satisfying \cref{assume:loss}.  The adaptive calibration strategy provided in \cref{subsec:dynamic_update} is a practical procedure that aims to approximately maintain this condition in a data-driven manner during optimization.
\end{remark}

\begin{remark}
Since the required network width in \cref{thm:conv} depends on the unknown  smoothness level $\beta$, the resulting procedure is not fully data-adaptive. Nevertheless, extending this result to adaptive inference procedures can be done without much difficulty by using recently proposed adaptive deep learning approaches \citep[e.g.,][]{ohn2022nonconvex,kong2023masked,kurisu2025adaptive}.
\end{remark}

Our analysis can be extended in a straightforward manner to a multiple-object setting. Suppose that the true decision region is represented as
    \begin{align*}
        \Gamma_\star=\bigcup_{k=1}^K  \Gamma_{\star,k}
    \end{align*}
where $\Gamma_{\star,1},\dots,\Gamma_{\star,K}$ are disjoint subsets of $[0,1]^2$, corresponding to $K$ objects in the image, and each $\Gamma_{\star,k}$ is of the form
\eqref{eq:true_region}. Even in this multiple-object setting, the proposed procedure attains the same convergence rate as in \cref{thm:conv}. Since the proof follows by applying the same argument to each component and summing the resulting errors, we omit the details.

\section{Numerical studies}
\label{sec:experiments}

\subsection{Simulation examples}
Under the assumption that the intensity inside the true region is higher than that outside, we conducted simulation studies under both Gaussian and Poisson noise models. 
All synthetic images were generated on a $128\times128$ lattice and rescaled to the continuous domain $[0,1]^2$.  For each scenario, we first generated a noiseless piecewise-constant image and then added Gaussian or Poisson noise.  In the Gaussian setting, the observed intensity was generated as
\[
Y(x)\mid x\in\Gamma^\star \sim N(\mu_{\mathrm{in}},\sigma^2),
\qquad
Y(x)\mid x\notin\Gamma^\star \sim N(\mu_{\mathrm{out}},\sigma^2),
\]
where we fix the normal means $\mu_{\mathrm{in}}=2$ and $\mu_{\mathrm{out}}=1$, and vary the standard deviation $\sigma\in\{0.1,0.4,0.7,1.0,1.3\}$. In the Poisson setting, the observed intensity was generated as
\[
Y(x)\mid x\in\Gamma^\star \sim \mathrm{Poisson}(\lambda_{\mathrm{in}}),
\qquad
Y(x)\mid x\notin\Gamma^\star \sim \mathrm{Poisson}(\lambda_{\mathrm{out}}),
\]
where we fix the outside intensity $\lambda_{\mathrm{out}}=3$ and vary the inside intensity $\lambda_{\mathrm{in}}\in\{5,7,9,11,13\}$. We considered three single-object shapes and one multi-object configuration.

In the single-object experiments, the true object region $\Gamma^\star$ was generated using the following fixed shape parameters.  For the star-shaped case, we used center $(0.5,0.5)$, outer radius $0.33$, inner radius $0.14$, five points, and rotation $0$.  For the triangular case, we used center $(0.5,0.5)$, radius $0.3$, and rotation $0$. Also, one side of triangle is replaced by sine wave.
For the elliptic case, we used center $(0.5,0.5)$, semi-axis lengths $(0.20,0.12)$, and rotation $0.1$. The multi-object scenario was constructed as the union of these three components.  We used the following fixed configuration: a star centered at $(0.35,0.35)$ with outer radius $0.18$, inner radius $0.08$, $5$ points, and rotation $0.2$. A triangle centered at $(0.72,0.35)$ with radius $0.10$ and rotation $1.0$, and an ellipse centered at $(0.55,0.72)$ with semi-axes $(0.14,0.09)$ and rotation $0.6$.

The decision function $g$ was trained using the continuous loss introduced in \eqref{eq:loss_update}
with $(\xi,\kappa,\tau)$ updated dynamically during optimization according to the procedure described in Section \eqref{subsec:dynamic_update}. And, the initial values of $\kappa$ and $\tau$ were set to $0.5$. The decision function was modeled by a three-layer neural network with ReLU activation function and measuring error with Lebesgue measure of the symmetric difference. We compared the proposed method with the methods of \citet{syring2020robust}, \citet{lajili2025unsupervised}, \citet{chan2001active}, and \citet{wang2023cut}. Each experiment was repeated five times, and we report the average symmetric difference error.

\begin{figure}[!htbp]
    \centering
    \includegraphics[width=1\linewidth]{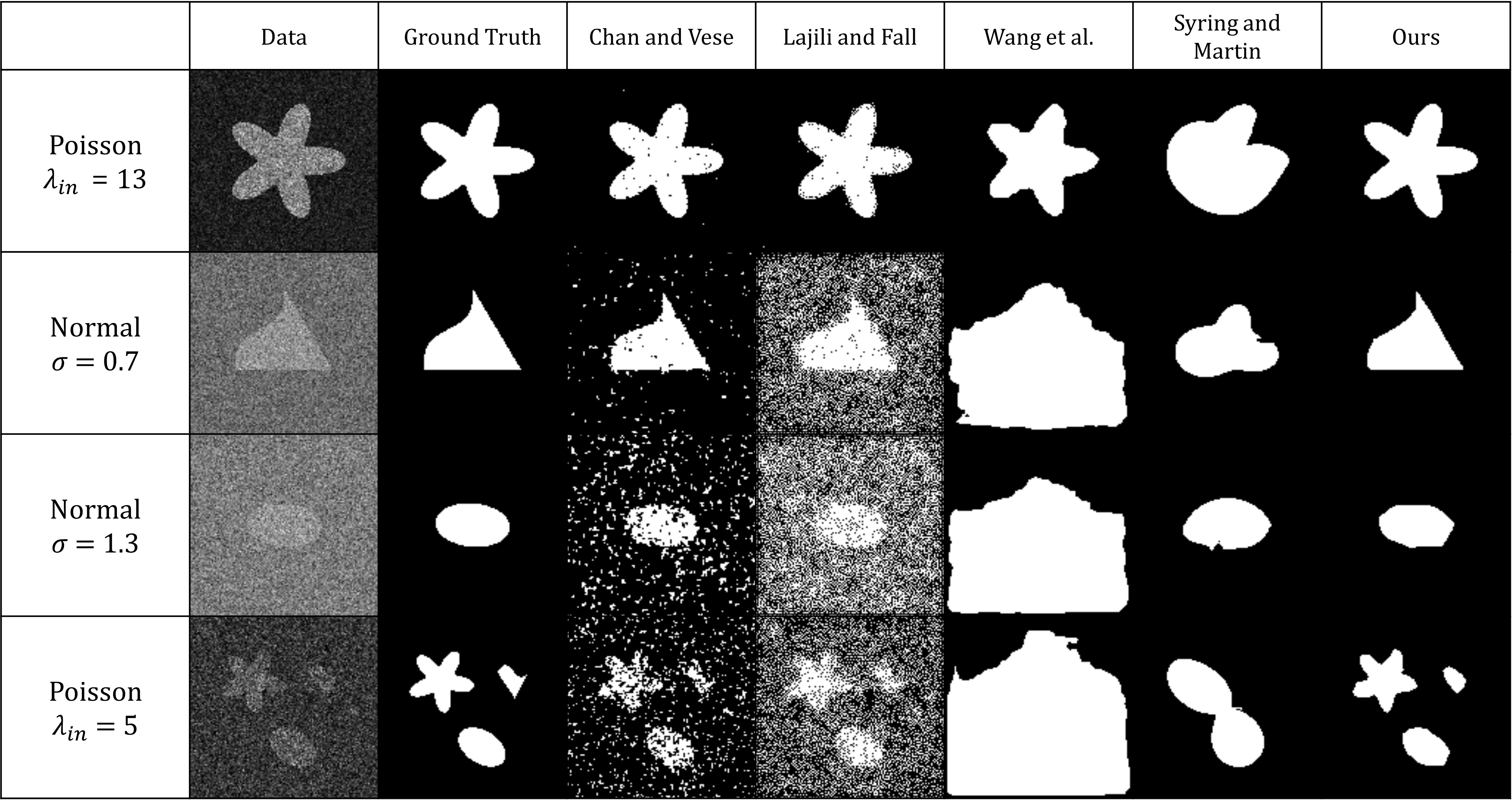}
    \caption{Qualitative comparison of boundary detection results in representative synthetic scenarios. Each row corresponds to a different object shape and noise setting, and each column shows the observed noisy image, the ground truth region, and the estimated regions produced by the competing methods. The proposed method accurately recovers the main object boundaries across both single- and multi-object settings, whereas competing methods often produce fragmented, distorted, or overly enlarged regions in more challenging cases.}
    \label{fig:fig1}
\end{figure}

Fig.~\ref{fig:fig1} presents qualitative visualizations of the competing methods for several representative scenarios. The proposed method faithfully recovers object boundaries across a variety of scenarios, whereas the competing methods often produce noisy boundaries, miss parts of the object, or detect distorted regions under more challenging noise conditions.

\begin{figure}[!htbp]
    \centering
    \includegraphics[width=1\linewidth]{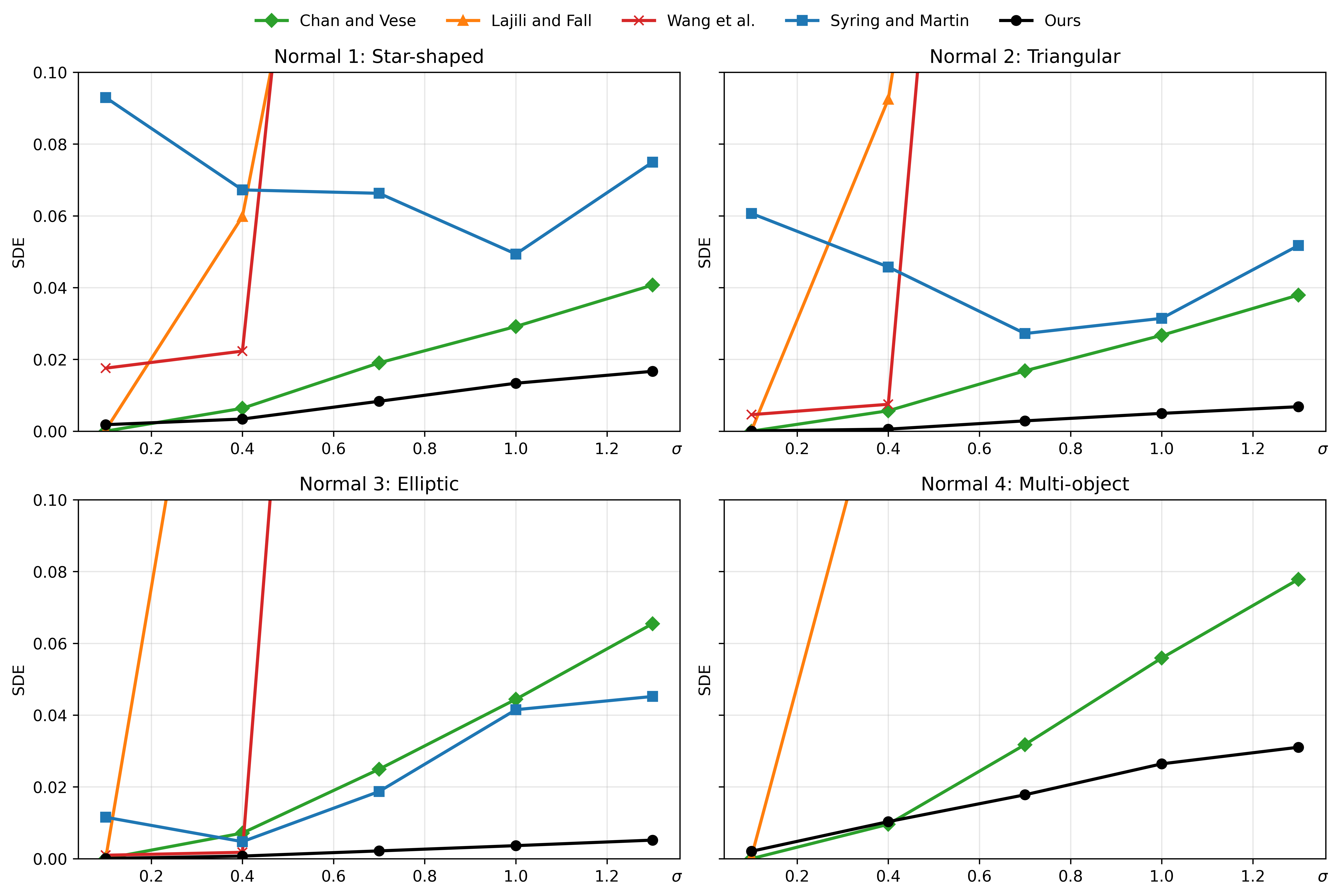}
    \caption{Symmetric difference error (SDE) under Gaussian noise for the synthetic examples. The error is measured by the Lebesgue measure of the symmetric difference  between the estimated region and the true object region, averaged over five independent repetitions. The four panels correspond to the star-shaped, triangular, elliptic,and multi-object scenarios. As the noise level $\sigma$ increases, the proposed method remains comparatively stable and shows a slower increase in error than the competing methods.}
    \label{fig:fig2}
\end{figure}

Fig.~\ref{fig:fig2} presents the  results under Gaussian noise, where the proposed method consistently achieved low error across all three single-object shapes. As the interior intensity contrast weakened and $\sigma$ increased, the performance of the competing methods deteriorated rapidly, while the proposed method showed only a gradual increase in error. A similar trend was observed in the multi-object setting. Although the error of our method also increased as the noise level became higher, its increase was much more moderate than that of the competing methods.

\begin{figure}[!htbp]
    \centering
    \includegraphics[width=1\linewidth]{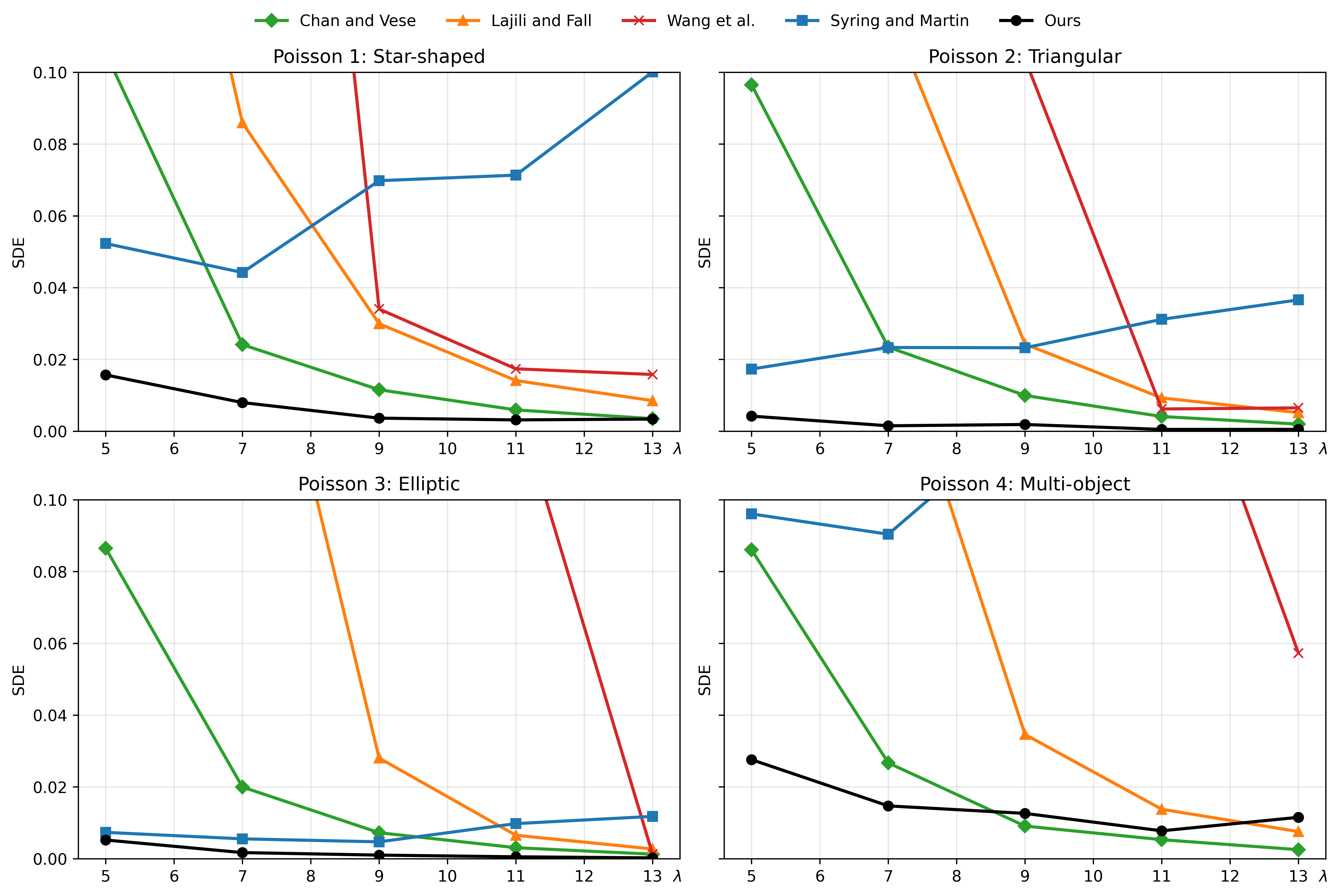}
    \caption{Symmetric difference error (SDE) under Poisson noise for the synthetic examples. The error is measured by the Lebesgue measure of the symmetric difference between the estimated and true regions, averaged over five independent repetitions. The four panels correspond to the star-shaped, triangular, elliptic, and multi-object scenarios. Smaller values of $\lambda_{\mathrm{in}}$ correspond to weaker separation between the object and the background, making boundary recovery more difficult. The proposed method maintains relatively low error across most scenarios.}
    \label{fig:fig3}
\end{figure}

A similar pattern was observed under Poisson noise, as shown in Fig.~\ref{fig:fig3}. The proposed method and the competing methods generally achieved low error in large $\lambda$. Even when the interior intensity parameter decreased so that distinguishing the boundary from the background became more difficult, the proposed method maintained low error in most scenarios. In contrast, the competing methods showed a rapid increase in error as $\lambda$ decreased, and in some cases failed to recover the boundary properly. In the low $\lambda$ settings, the error of the proposed method also increased under the most challenging low-intensity conditions, but overall it still exhibited stable performance.

\subsection{Real data analysis}
 
We also considered real microscopy images from the BBBC038 dataset, a publicly available benchmark dataset from the Broad Bioimage Benchmark Collection \citep{caicedo2019nucleus}. The dataset consists of fluorescence microscopy images of cell nuclei and their corresponding annotated pseudo masks. 
Compared with the synthetic settings, these images are considerably more heterogeneous: nuclei vary in shape, size, contrast, and local background intensity. This makes the dataset a useful testbed for examining whether a boundary detection method remains stable under realistic image conditions. The dataset gives truth nucleus mask, respectively. For evaluation, the masks contained in each nucleus were merged into a whole ground truth mask, which was treated as the true object region $\Gamma^\star$.

We selected $13$ images from the dataset and randomly extracted five $64\times64$ crops from each image, giving $65$ cropped images in total.  The crops were taken from different spatial locations so that the experiment included a range of local nuclear structures and background patterns.  All cropped images were normalized to have intensity values in $[0,1]$, and the pixel locations were rescaled to $[0,1]^2$. 

To assess robustness to additional noise, we added independent Gaussian noise to the normalized images.  For each pixel location $x$, the noisy intensity was generated by
\[
Y_\sigma(x)=Y_0(x)+\varepsilon(x),
\qquad
\varepsilon(x)\sim N(0,\sigma^2),
\]
where $Y_0(x)$ is the normalized original intensity and
\[
\sigma\in\{0.00,0.02,0.04,0.06,0.08,0.10\}.
\]
The resulting noisy images were clipped to the range $[0,1]$. 

The proposed method was compared with the methods of \citet{syring2020robust}, \citet{lajili2025unsupervised}, \citet{chan2001active}, and \citet{wang2023cut}, which were also used in the simulation experiments.  The error was measured by the Lebesgue measure of the symmetric difference between the estimated region and the pseudo mask on the $64\times64$ grid, and the errors were averaged over the $65$ crops for each noise level.

\begin{figure}[h]
    \centering
    \includegraphics[width=1\linewidth]{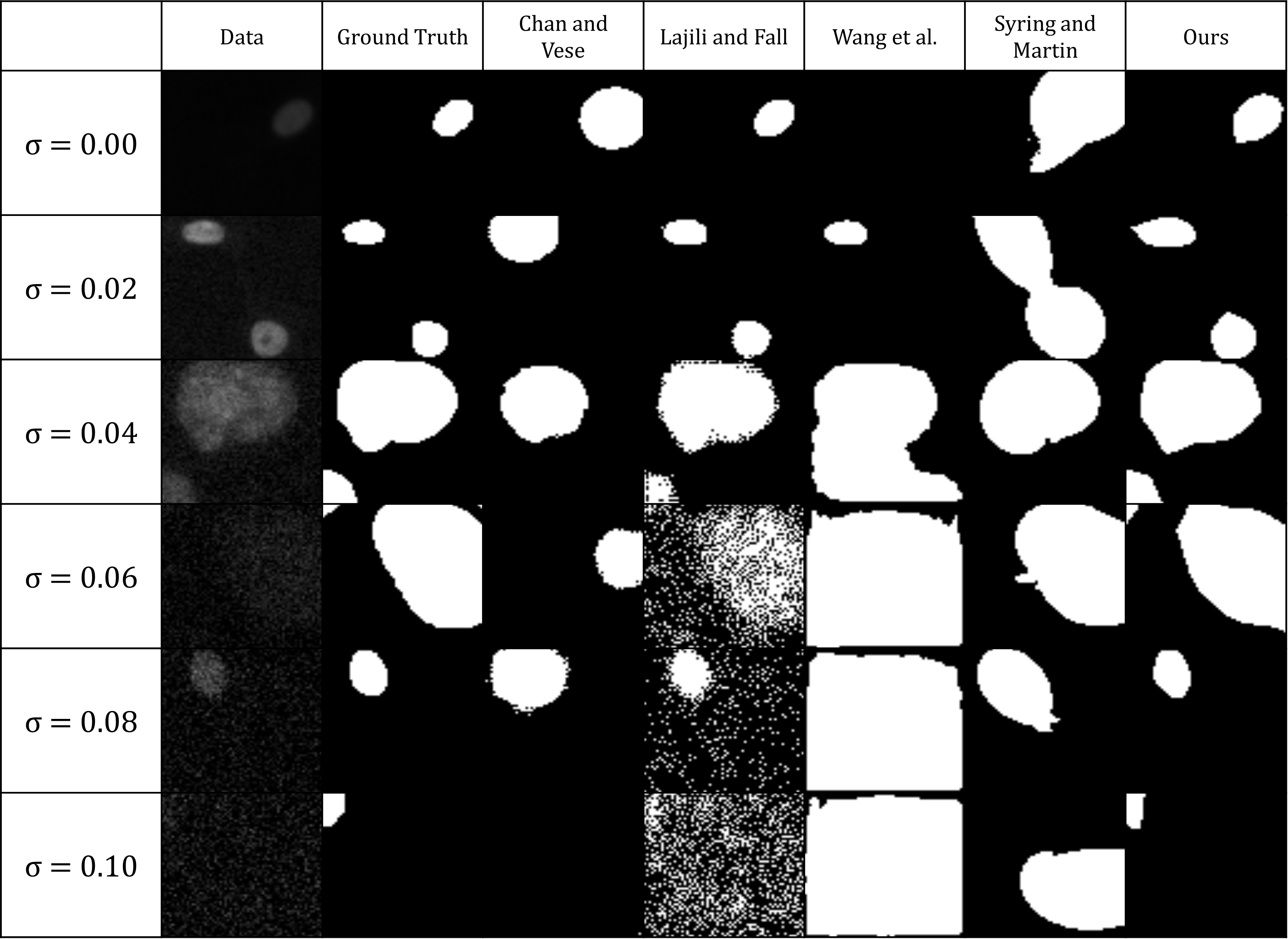}
    \caption{Qualitative comparison on cropped BBBC038 microscopy images under different levels of added Gaussian noise. Each row corresponds to a different noise level, and each column shows the noisy input image, the ground-truth mask, and the estimated regions obtained by the competing methods. The results illustrate the difficulty of real microscopy images, where nuclei may have weak, blurred, or spatially heterogeneous boundaries. The proposed method remains comparatively stable under moderate and high noise levels.}
    \label{fig:bbbc_table}
\end{figure}

Fig.~\ref{fig:bbbc_table} shows representative segmentation results.  The real microscopy images are more difficult than the synthetic examples because the foreground and background are not separated by a simple intensity jump. In several crops, the nuclei have weak or blurred boundaries, and the surrounding regions show gradual intensity changes.  Some nuclei are also truncated by the crop boundary. These features make the problem less consistent with an ideal two-region model. This mismatch is most visible in the noise-free case. The proposed method assumes that the image can be approximately described by two intensity distributions, separated by a single target region $\Gamma_\star$.  However, in some BBBC038 crops, the nuclear signal has a layered structure: a bright core is surrounded by a weaker gradation region. In such cases, the adaptive calibration step may regard part of this gradation as belonging to the interior distribution $\Fin$, which leads to overly large estimates of object regions.

After additional Gaussian noise is introduced, this local gradation becomes less dominant, and the proposed method tends to recover the main nuclear region more consistently.  The baseline methods show different types of instability as the noise level increases. The same phenomenon is observed in \citep{syring2020robust} The method of \citet{lajili2025unsupervised} often produces speckle-like false positives, whereas \citet{wang2023cut} tends to estimate overly large object regions in high-noise cases.  The method of \citet{chan2001active}, initialized from a circle around the brightest region, is relatively stable in some crops but usually fails to converge to the desired boundary when the intensity is weak or the image is noisy.

\begin{figure}[h]
    \centering
    \includegraphics[width=0.75\linewidth]{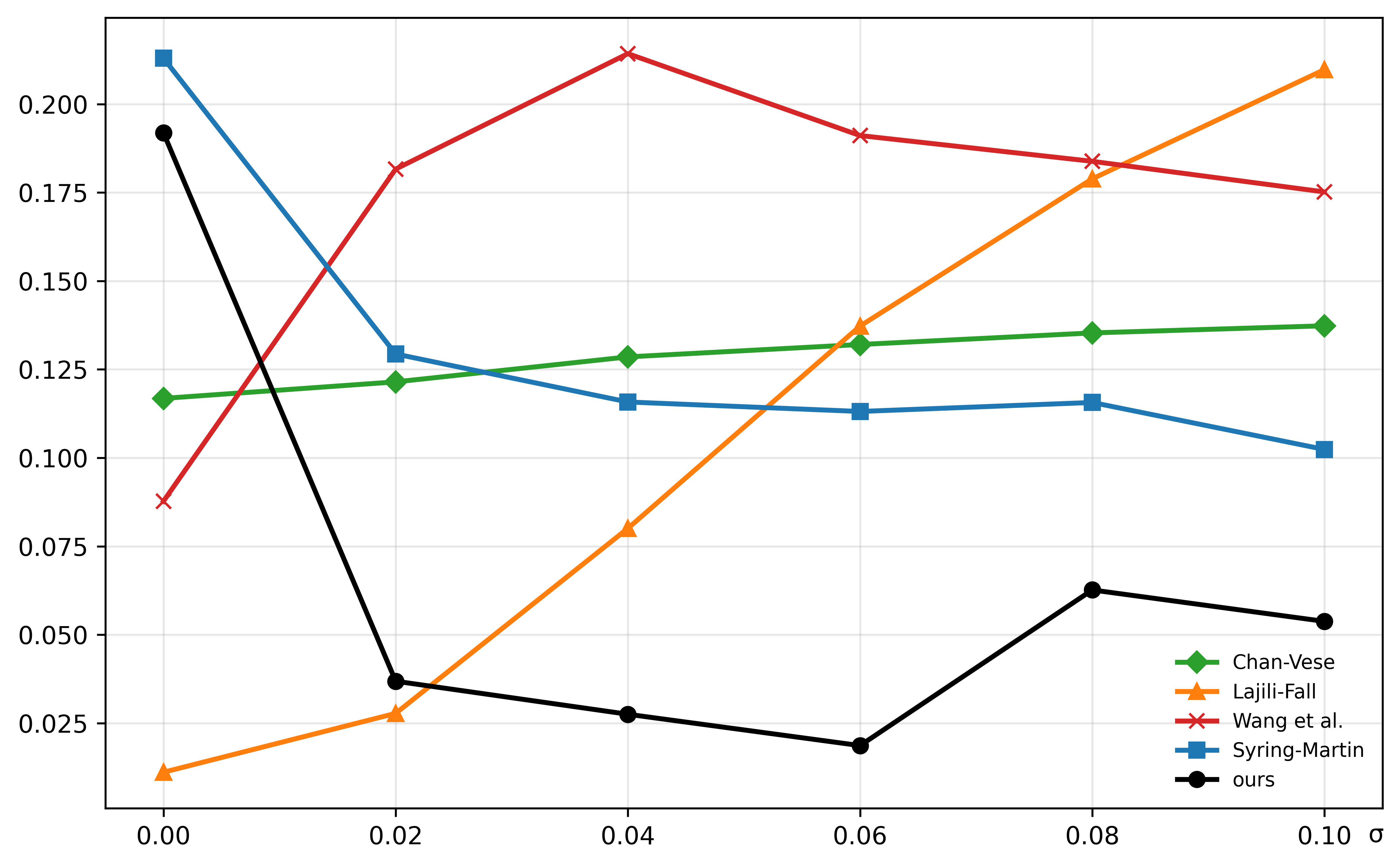}
    \caption{Symmetric difference error on cropped BBBC038 microscopy images under different Gaussian noise levels. The error is computed between the estimated region and the pseudo ground-truth mask and averaged over 65 cropped images. The proposed method performs less favorably in the noise-free case but yields lower error than the competing methods for most positive noise levels.}
    \label{fig:bbbc_error}
\end{figure}

The quantitative results are shown in Fig.~\ref{fig:bbbc_error}.  At $\sigma=0.00$, the method of \citet{lajili2025unsupervised} gives the smallest error, while the proposed method performs poorly on some crops. This is consistent with the qualitative observation that the noise-free BBBC038 images may violate the simple two-region intensity assumption. However, the error of our method decreases substantially once moderate noise is added and remains lower than those of the competing methods for most positive noise levels. We would emphasize that this behavior should not be interpreted as a general benefit of adding noise. Rather, it reflects a mismatch between the two-region intensity model and the clean BBBC038 crops, where nuclei often exhibit gradual intensity transitions. Adding moderate Gaussian noise partially suppresses such local gradation effects, making the effective two-region approximation more appropriate for the proposed calibration scheme.

\section{Concluding remarks}
\label{sec:conclusion}

We have developed a continuous-loss framework for boundary detection from unlabeled noisy images. By relaxing the discontinuous threshold-based misclassification loss of \citet{syring2020robust} into a hinge-type surrogate, we obtain an objective that is amenable to gradient-based optimization with deep neural networks while retaining the statistical identifiability of the true object region. We have established Fisher consistency and a calibration inequality for the proposed loss, and have shown that the corresponding deep network estimator attains the minimax rate $n^{-\beta/(\beta+1)}$ over a piecewise smooth boundary class that strictly generalizes the globally smooth models considered in previous work.

Our results suggest that the apparent gap between the combinatorial loss formulations used in the statistical literature on boundary inference and the continuous surrogate losses used in modern deep learning is not fundamental: a carefully designed continuous surrogate can recover the same rate-optimal guarantees without sacrificing scalability. We view this as a step toward bringing modern empirical risk minimization tools into statistically principled boundary inference.

Several directions remain open. First, our convergence theory is established for fixed loss parameters satisfying the separation condition in \cref{assume:loss}, whereas the adaptive calibration approach of \cref{subsec:dynamic_update} updates these loss parameters along the optimization trajectory. Providing a theoretical guarantee for the adaptive scheme is an important next step. Second, the present formulation treats the observed intensity $Y$ as a scalar, which corresponds to grayscale imaging. Many practical imaging modalities are inherently multichannel such as RGB photographs and multispectral satellite imagery, which produce vector-valued intensities at each pixel. The proposed framework extends naturally to this setting once a scalar score summarizing the channel-wise contrast is specified. Concretely, we can replace the scalar threshold rule $u_y = 2\ind(y > \xi) - 1$ in \eqref{eq:pseudo_lable} by a pseudo-label of the form $u_y = 2\ind(s(y) > \xi) - 1$ for some learned or prescribed score function $s$, while keeping the hinge-type loss and the adaptive calibration step unchanged. Empirical and theoretical analyses of this approach are left for future work.

\section*{Acknowledgement}
This work was supported by the National Research Foundation of Korea (NRF) funded by the Korea government (MSIT) (RS-2024-00411853) and INHA UNIVERSITY Research Grant.

\appendix
\section{Proofs}
\label{sec:proofs}

\subsection{Proof of \cref{prop:fisher_consist}}

\begin{proof}
For simplicity, let $u:=u_y$. The difference between  loss values of $g$ and the true $g_\star$ is given by
    \begin{align*}
    d_g(y,x)
        &:=\ell_g(y,x)-\ell_{g_\star}(y,x)\\
        &=(1-ug(x))_+\{\kappa(u+1)/2+\tau(1-u)/2\}\\
        &\quad -(1-ug_\star(x))_+\{\kappa(u+1)/2+\tau(1-u)/2\}\\
        &=\kappa\ind(u=1)\{(1-g(x))_+-(1-g_\star(x))_+\}\\
         &\quad +\tau\ind(u=-1)\{(1+g(x))_+-(1+g_\star(x))_+\}
    \end{align*}
Since we restrict our attention to functions bounded by 1, we have
    \begin{align}
    \label{eq:loss_diff}
    d_g(y,x)
        &=\{-\kappa\ind(u=1)+\tau\ind(u=-1)\}\{g(x)-g_\star(x)\}
    \end{align}
We aim to find $g$ that minimizes the expectation of the above.
We deal with the following two cases separately.
\begin{itemize}
    \item For $x$ such that $g_\star(x)=1$, we have $\E(u=1|x)=\P(Y>\xi|x)=1-\Fin(\xi)$. For such $x$, we have 
         \begin{align}
         \E[d_g(Y,X)|X=x]
      &=\{-\kappa(1-\Fin(\xi))+\tau \Fin(\xi)\} (g(x)-1)\nonumber\\
       &= \{\kappa-(\kappa+\tau)\Fin(\xi)\}(1-g(x)) \label{eq:mean_loss_diff_1}.
        \end{align}
            By assumption that $\kappa/(\kappa+\tau)> \Fin(\xi)$, the above is minimized at $g(x)=1=g_\star(x)$.
    \item  For $x$ such that $g_\star(x)=-1$, we have $\E(u=1|x)=\P(Y>\xi|x)=1-\Fout(\xi)$.
    For such $x$, we have 
        \begin{align}
        \E[d_g(Y,X)|X=x] 
        & =\{ \tau \Fout(\xi)-\kappa(1-\Fout(\xi))\}( 1+g(x))\nonumber\\
      & =\{ (\tau+\kappa) \Fout(\xi)-\kappa\}( 1+g(x)) \label{eq:mean_loss_diff_2}.
        \end{align}
    By assumption that $\Fout(\xi)>\kappa/(\kappa+\tau)$, the above is minimized at $g(x)=-1=g_\star(x)$.
\end{itemize}
These considerations reveal that the conditional expectation $\E[d_g(Y,X)|X=x]$ is minimized at $g_\star(x)$ always, which completes the proof.
\end{proof}

\subsection{Proof of \cref{prop:calib}}

\begin{proof}
In the proof of \cref{prop:fisher_consist}, we have shown that
    \begin{align}
    \label{eq:lower_bound}
        \E[\ell_g(y,x)-\ell_{g_\star}(y,x)|X=x]
        \ge \tC' |g_\star(x)-g(x)|,
    \end{align}
where the constant $\tC'$ is given by
    \begin{align*}
       \tC':= \min\cbr{\kappa-(\kappa+\tau)\Fin(\xi), (\tau+\kappa) \Fout(\xi)-\kappa}
    \end{align*}
and is positive due to   \cref{assume:loss}.
Taking the expectation to both sides of \eqref{eq:lower_bound} with respect to $\Q$, we have
      \begin{align*}
        \cR(g)-\cR(g_\star)\ge \tC'\int |g(x)-g_\star(x)|\d\Q(x)
    \end{align*}
Moreover, by \cref{assume:pixel}, we have
    \begin{align*}
      \int |g(x)-g_\star(x)|\d\Q(x)\ge \frac{1}{A}\|g-g_\star\|_{\cL_1([0,1]^2)}.
    \end{align*}    
Therefore, the desired result follows from that
    \begin{align*}
        \lambda(\Gamma \triangle \Gamma_\star)
     & =  \int \ind(g(x)<0,g_\star(x)=1) \d \lambda(x)
     + \int \ind(g(x)\ge0,g_\star(x)=-1) \d \lambda(x)\\
     &\le \int |1-g(x)| \ind(g_\star(x)=1) \d \lambda(x)
     +\int |-1-g(x)| \ind(g_\star(x)=-1) \d \lambda(x)\\
     &=\int |g_\star(x)-g(x)|\d \lambda(x)=\|g-g_\star\|_{\cL_1([0,1]^2)}.
    \end{align*}
    
\end{proof}

\subsection{Proof of \cref{thm:approximation}}

\begin{proof}
By Theorem A.1 of \citet{fang2024intrinsic}, there exists a neural network $h^\dag_j\in \cG(L_0', \tC_1' D, \tC_1 'D^{\tC_2'})$ such that
    \begin{align*}
        \|h^\dag_j-h_j\|_{\cL_\infty([0,1])} \le \epsilon:=\tC_3'D^{-2\beta}.
    \end{align*}
for some positive constants $L_0'$, $\tC_1'$, $\tC_2'$ and $\tC_3'$. Without loss of generality, we assume $K_j =\{x: x_1\ge h_j(x_2)\}$ as the proofs for the other cases are exactly the same. We then construct a neural network $k^\dag_j$ as 
    \begin{align*}
        k^\dag_j (x)&= \frac{1}{2}\cbr{\rho\del[1]{\epsilon^{-1}(x_1-h^\dag_j(x_2)+\epsilon)}-\rho\del[1]{\epsilon^{-1}(x_1-h^\dag_j(x_2)-\epsilon)}}\\
        &=
        \begin{cases}
            1 & \text{ if $x_1\ge h^\dag_j(x_2)+\epsilon$}\\
            0 & \text{ if $x_1<h^\dag_j(x_2)-\epsilon$}\\
             \frac{1}{2}\epsilon^{-1}(x_1-h^\dag_j(x_2))+\frac{1}{2} & \text{ otherwise}
        \end{cases}.
    \end{align*}
Therefore, since $h_j(x_2)-\epsilon\le h^\dag_j(x_2)\le h_j(x_2)+\epsilon$ for any $x$,
    \begin{align*}
        \int & |k^\dag_j (x)-\ind(x_1\ge h_j(x_2))|\d\lambda(x)\\
        &= \int_{\{x:x_1\ge h_j(x_2)\}}  |k^\dag_j (x)-1|\d\lambda(x)
        +\int_{\{x:x_1< h_j(x_2)\}} |k^\dag_j (x)|\d\lambda(x)\\
        &\le \int_{\{x:x_1\ge h^\dag_j(x_2)-\epsilon\}}  |k^\dag_j (x)-1|\d\lambda(x)
        +\int_{\{x:x_1< h^\dag_j(x_2)+\epsilon\}} |k^\dag_j (x)|\d\lambda(x)\\
         &=  \int_{\{x: h^\dag_j(x_2)+\epsilon\ge x_1\ge h^\dag_j(x_2)-\epsilon\}} \frac{1}{2}\abs{\epsilon^{-1}(x_1-h^\dag_j(x_2))-1} \d\lambda(x)\\
       &\quad  +\int_{\{x:h_j(x_2)-\epsilon\le x_1< h^\dag_j(x_2)+\epsilon\}}\frac{1}{2}\abs{\epsilon^{-1}(x_1-h^\dag_j(x_2))+1} \d\lambda(x)\\
       &=\int_0^1 \int_{h^\dag_j(x_2)-\epsilon}^{h^\dag_j(x_2)+\epsilon}\frac{1}{2}\abs{\epsilon^{-1}(x_1-h^\dag_j(x_2))-1} \d x_1 \d x_2\\
       &\quad +\int_0^1 \int_{h^\dag_j(x_2)-\epsilon}^{h^\dag_j(x_2)+\epsilon}\frac{1}{2}\abs{\epsilon^{-1}(x_1-h^\dag_j(x_2))+1} \d x_1 \d x_2\\
       &\le 2\epsilon.
    \end{align*}
Finally, invoking Lemma 16 of \citet{fan2024factor}, which provides an accurate neural network approximation of the multiplication operation, we obtain
    \begin{align*}
       \| g^\dag - g_\star\|_{\cL_1([0,1]^2)}
       &\le  \norm{ g^\dag - \prod _{j=1}^J k_j^\dag}_{\cL_1([0,1]^2)}
       +\norm{  \prod _{j=1}^J k_j^\dag - g_\star}_{\cL_1([0,1]^2)}\\
       &\le D^{-2\beta}+\sum_{j=1}^J\norm{k^\dag_j (x)-\ind(x_1\ge h_j(x_2))}_{\cL_1([0,1]^2)}\\
       &\le D^{-2\beta}+2J\epsilon
    \end{align*}
for some neural network $g^\dag$ whose architecture satisfies the required condition.
\end{proof}

\subsection{Proof of \cref{thm:conv} }

Before proceeding to the proof, we provide one technical lemma.

\begin{lemma}
\label{lemma:variance}
Under \cref{assume:loss}, there exists an absolute constant $\tC_0>0$ such that
    \begin{align*}
        \E[(\ell_g(Y,X)-\ell_{g_\star}(Y,X))^2]\le \tC_0 \{\cR(g)-\cR(g_\star)\}.
        \end{align*}
for any $g\in\bcG.$
\end{lemma}

\begin{proof}
Recall the identity we have established in the proof of \cref{prop:fisher_consist}:
  \begin{align*}
       d_g(y,x)
       &:=\ell_g(y,x)-\ell_{g_\star}(y,x)\\
        &=\{-\kappa\ind(u=1)+\tau\ind(u=-1)\}\{g(x)-g_\star(x)\}
    \end{align*}
where $u:=u_y$. Note that $|g(x)|\le 1$ for any $x$. For $x$ such that $g_\star(x)=1$, we have 
         \begin{align*}
         \E[ d_g(Y,X)^2|X=x]
         &=\E[\kappa^2\ind(u=1)+\tau^2\ind(u=-1)](g(x)-1)^2\\
         &\le 2\E[\kappa^2\ind(u=1)+\tau^2\ind(u=-1)](1-g(x))\\
         &= \frac{2\E[\kappa^2\ind(u=1)+\tau^2\ind(u=-1)]}{\kappa-(\kappa+\tau)\Fin(\xi)}\E[d_g(Y,X)|X=x],
       \end{align*}
where we use the equality \eqref{eq:mean_loss_diff_1} established in the proof of \cref{prop:fisher_consist} for the last equality.
Similarly, for $x$ such that $g_\star(x)=-1$, we have 
         \begin{align*}
         \E[ d_g(Y,X)^2|X=x]
         &=\E[\kappa^2\ind(u=1)+\tau^2\ind(u=-1)](g(x)+1)^2\\
         &\le 2\E[\kappa^2\ind(u=1)+\tau^2\ind(u=-1)](g(x)+1)\\
         &=  \frac{2\E[\kappa^2\ind(u=1)+\tau^2\ind(u=-1)]}{(\tau+\kappa) \Fout(\xi)-\kappa}\E[d_g(Y,X)|X=x],
        \end{align*}
where we use the equality \eqref{eq:mean_loss_diff_2} for the last equality.
As these inequalities uniformly hold, we get the desired result.
\end{proof}

\begin{proof}[Proof of of \cref{thm:conv}]
We first introduce some additional notation. Let $\cG_n:=\cG(L_0, M_n, B_n)$. Let $\hcR(g):=n^{-1}\sum_{i=1}^n \ell_g(Y_i,X_i)$ denote the empirical risk of a function $g$. Moreover, we denote the excess population and empirical risks of $g$ by
    \begin{align*}
        \cE(g):=\cR(g)-\cR(g_\star),
        \quad
        \hcE(g):=\hcR(g)-\hcR(g_\star),
    \end{align*}
respectively. We begin with the inequality
    \begin{align}
        \cR(\hg_n)-\cR(g_\star)
        &=2\{\hcR(\hg_n)-\hcR(g_\star)\} + \cR(\hg_n)-\cR(g_\star)-2\{\hcR(\hg_n)-\hcR(g_\star)\}\nonumber\\
        &\le 2\{\hcR(\hg_n)-\hcR(g_\star)\}+\sup_{g\in\cG_n}\{\cE(g)-2\hcE(g)\}.\label{eq:basic_ineq}
    \end{align}
By the optimization optimality of $\hg_n$ and the Lipschitzness of the loss function, the first term in \eqref{eq:basic_ineq} is bounded as
    \begin{align*}
       \hcR(\hg_n)-\hcR(g_\star)
       \le \hcR(g^\dag)-\hcR(g_\star)
       \lesssim \frac{1}{n}\sum_{i=1}^n|g^\dag(X_i)-g_\star(X_i)|,
    \end{align*}
where $g^\dag$ is a neural network that closely approximates $g_\star$, whose existence is guaranteed by \cref{thm:approximation}. By \cref{assume:pixel}, we have $ \E[|g^\dag(X_i)-g_\star(X_i)|]\lesssim \|g^\dag-g_\star\|_{\cL_1([0,1]^2)}$ and
    \begin{align*}
        \var(|g^\dag(X_i)-g_\star(X_i)|)\le \E[|g^\dag(X_i)-g_\star(X_i)|^2]&\lesssim \|g^\dag-g_\star\|_{\cL_1([0,1]^2)}.
    \end{align*}
where the second inequality follows from that $|g^\dag(x)-g_\star(x)|\le 2$ for any $x\in[0,1]^2$. Therefore, by Bernstein's inequality, we have
    \begin{align}
    \label{eq:term1}
         \frac{1}{n}\sum_{i=1}^n|g^\dag(X_i)-g_\star(X_i)|
         \lesssim  \|g^\dag-g_\star\|_{\cL_1([0,1]^2)}+\frac{\log(1/\delta)}{n}
    \end{align}
with probability at least $1-\delta/2$. For the second term in \eqref{eq:basic_ineq}, we consider a minimal $n^{-1}$-net $\{g_j^\circ:j\in[N]\}$ of the function space $\cG_n$ with respect to the $\cL_\infty$-norm.
Then by the Lipschitzness of the loss function, we have
    \begin{align*}
        \sup_{g\in\cG_n}\{\cE(g)-2\hcE(g)\}\lesssim  \frac{1}{n} +\max_{1\le j\le N}\{\cE(g_j^\circ)-2\hcE(g_j^\circ)\}.
    \end{align*}
For $g\in\bcG$, define a random variable
    \begin{align*}
        Z^g_i:=\E[d_g(Y,X)]-d_g(Y_i,X_i) \text{ with }d_g(Y_i,X_i):=\ell_g(Y_i,X_i)-\ell_{g_\star}(Y_i,X_i)
    \end{align*}
so that we write $\cE(g)-\hcE(g)=n^{-1}\sum_{i=1}^nZ^g_i$. Then for any $g\in\bcG$ and $i\in[n]$, we have $|Z^g_i|\le A_1:=8\max\{\kappa,\tau\}$, i.e., $Z^g_i$ is a bounded random variable. Moreover, we have $\var(Z^g_i)\le A_2\cE(g)$ for some constant $A_2>0$ by \cref{lemma:variance}. Then by Bernstein's inequality, we have
    \begin{align*}
      \P\del{\cE(g)-2\hcE(g)\ge t}  
      &=\P\del{\cE(g)-\hcE(g)\ge \frac{t}{2}+\frac{1}{2}\cE(g)} \\
      &=\P\del{\sum_{i=1}^nZ^g_i\ge  \frac{nt}{2}+\frac{n}{2}\cE(g)}\\
     & \le \exp\del{-\frac{n^2(t+\cE(g))^2/8}{n\sum_{i=1}^n\var(Z^g_i)+nA_1t/3}}\\
       & \le \exp\del{-\frac{n(t+\cE(g))^2/8}{A_2\cE(g)+A_1t/3}}\\
     &\le \exp\del{-\tC_1'n(t+\cE(g))}\le \exp(-\tC'_1nt)
    \end{align*}
for some constant $\tC_1'>0$. Hence, by the union bound
    \begin{align*}
        \P\del{\max_{1\le j\le N}\{\cE(g_j^\circ)-2\hcE(g_j^\circ)\}\ge t}
        &\le \sum_{j=1}^N\P\del{\cE(g_j^\circ)-2\hcE(g_j^\circ)\ge t}\\
        &\le \exp(\log N-\tC_1'nt).
    \end{align*}
Thus, by taking $t=\{\log N+\log(2/\delta)\}/(\tC_1'n)$, we have
    \begin{align}
     \label{eq:term2}
         \sup_{g\in\cG_n}\{\cE(g)-2\hcE(g)\}
         \lesssim \|g^\dag-g_\star\|_{\cL_\infty([0,1]^2)}+ \frac{\log N}{n} +\frac{\log(1/\delta)}{n}
    \end{align}
with probability at least $1-\delta/2$. By putting together the last display and \eqref{eq:term1}, we obtain
     \begin{align}
         \cR(\hg_n)-\cR(g_\star)
         \lesssim \|g^\dag-g_\star\|_{\cL_\infty([0,1]^2)}+ \frac{\log N}{n} +\frac{\log(1/\delta)}{n}
    \end{align}
with probability at least $1-\delta$. Next, we apply \cref{thm:approximation} with the choice $D\asymp n^{1/(2\beta+2)}$. Then by the well-known bound of the log covering number of a neural network function class \citep[e.g., Lemma K.4 of][]{ohn2024adaptive}, we have $\log(N)\lesssim D^2\log n\asymp n^{1/(\beta+1)}\log n$ . Thus, we have
  \begin{align*}
         \cR(\hg_n)-\cR(g_\star)
         \lesssim n^{-\beta/(\beta+1)}\log n+\frac{\log(1/\delta)}{n}
    \end{align*}
with probability at least $1-\delta$. The calibration inequality in \cref{prop:calib} concludes the proof.
\end{proof}

\bibliographystyle{plainnat}
\bibliography{_references}

\end{document}